\documentclass[12pt]{article}

\usepackage{amssymb,latexsym}
\usepackage{graphicx}
\usepackage{pgf,pgfarrows,pgfnodes,pgfautomata,pgfheaps,pgfshade}
\usepackage{tikz}
\usetikzlibrary{calc}
\usepackage{color}
\usepackage{paralist}
\usepackage[]{units}

\usepackage{amsmath,amsfonts,amsthm}

\allowdisplaybreaks

\newtheorem{theorem}{Theorem}
\newtheorem{proposition}{Proposition}
\newtheorem{corollary}[proposition]{Corollary}

\newtheorem{remark}{Remark}
\newtheorem{definition}{Definition}
\newtheorem{example}{Example}

\newcommand{\best}{{{\mathrm{top}}}}

\newcommand{\calR}{{{\mathcal{R}}}}

\newcommand{\np}{{{\mathrm{NP}}}}

\title{The Condorcet Principle for Multiwinner Elections: \\ From Shortlisting to Proportionality}
\author{Haris Aziz\\
Data61, CSIRO and UNSW\\
Sydney, Australia
\and
Edith Elkind\\
  University of Oxford\\
  Oxford, UK
\and
Piotr Faliszewski\\
AGH University\\
Krakow, Poland
\and
Martin Lackner\\
  University of Oxford\\
  Oxford, UK
  \and 
Piotr Skowron\\
  University of Oxford\\
  Oxford, UK
}
\date{}

\begin{document}

\maketitle

\begin{abstract}
We study two notions of stability in multiwinner elections that are based on the Condorcet criterion. The first notion was 
introduced by Gehrlein: A committee is stable if each committee member is preferred to each non-member by a (possibly 
weak) majority of voters. The second notion is called local stability (introduced in this paper): A size-$k$ committee is 
locally stable in an election with $n$ voters if there is no candidate $c$ and no group of more than $\frac{n}{k+1}$ 
voters such that each voter in this group prefers $c$ to each committee member. We argue that Gehrlein-stable 
committees are appropriate for shortlisting tasks, and that locally stable committees are better suited for applications 
that require proportional representation. The goal of this paper is to analyze these notions in detail, explore their 
compatibility with notions of proportionality, and investigate the computational complexity of related algorithmic tasks.
\end{abstract}

\section{Introduction}
\noindent
The notion of a Condorcet winner is among the most important ones in
(computational) social
choice~\cite{arrow2002handbook,handbook-comsoc}. Consider a group of agents,
each with a preference order over a given
set of candidates. The Condorcet condition says that if
there exists a candidate $c$ that is preferred to every other candidate
by a majority of agents (perhaps a different majority in each case),
then this candidate $c$ should be seen as the collectively best
option. Such a candidate is known as the Condorcet winner.

In single-winner elections, that is, in settings where the goal is to choose
one candidate (presidential elections are a prime example here), there are strong arguments
for choosing a Condorcet winner whenever it exists. For example, in case of presidential elections 
if a Condorcet winner existed but was not chosen as the country's
president, a majority of the voters might revolt. 
  (We note, however, that there are also arguments against rules that
  choose Condorcet winners whenever they exist: for example, such
  rules suffer from the no-show paradox~\cite{Moulin88,brandt2016optimal} 
  and fail the reinforcement axiom~\cite{RePEc:cup:cbooks:9780521360555}.)

In this paper, we consider multiwinner elections, that is, 
settings where instead of choosing a single winner (say, the president) 
we choose a collective body of a given size (say, a parliament).
The goal of our paper is to analyze generalizations of the
concept of a Condorcet winner to multiwinner elections. 
There are several natural definitions  of ``a Condorcet committee'' and we consider their merits and application
domains (we write ``Condorcet committee'' in quotes because several
notions could be seen as deserving this term and, thus, eventually we
do not use it for any of them).

First, we can take the approach of Gehrlein~\cite{Gehrlein1985199} and
Ratliff~\cite{ratliff2003some}, where we want the committee to be a
collection of high-quality individuals who do not necessarily need to
cooperate with each other (this is a natural approach, e.g., when we
are shortlisting a group of people for an academic position or for
some prize~\cite{bar-coe:j:non-controversial-k-names,elk-fal-sko-sli:c:multiwinner-rules}).
In this case, each member of the ``Condorcet committee'' should be
preferred by a majority of voters to all the non-members.

Alternatively, there is the approach of Fishburn~\cite{fishburn1981}
(also analyzed from an algorithmic perspective by Darmann~\cite{Darmann2013}), where
we assume that the committee members have to work so closely with each
other that it only makes sense to consider voters' preferences over
entire committees rather than over individual candidates (this would be
natural in selecting, e.g., small working groups). In this case, a ``Condorcet committee'' is
a committee that is preferred to every other committee by a majority of
voters. However, this approach is of limited use when voters express 
their preferences over individual candidates:
while such preferences can be lifted to preferences over committees, 
e.g., by using a scoring function, in the presence of strong synergies
among the committee members the induced preferences over committees
are unlikely to offer a good approximation of voters' true preferences.
Therefore, we do not pursue this approach in our work.

Finally, there is a middle-ground approach, proposed by Elkind et
al.~\cite{journals/scw/ElkindLS15}, where the committee members focus
on representing the voters (this is the case, e.g., in parliamentary
elections). In this case, we compare committees against single
candidates: We say that a voter $i$ prefers committee $W$ to some
candidate $c$ if there exists a candidate $w$ in $W$ (who can be seen as the
representative of this voter) such that $i$ prefers $w$ to $c$. Now we
could say that a ``Condorcet committee'' is one that is preferred to
each candidate outside the committee by a majority of voters. Indeed, Elkind et
al.~\cite{journals/scw/ElkindLS15} refer to such committees as
\emph{Condorcet winning sets}.

Elkind et al.~\cite{journals/scw/ElkindLS15} were unable to find
an election with no Condorcet winning set of size three; their empirical results
suggest that such elections are very unlikely.
Thus, to use their approach in order to select large committees in a meaningful way, 
we should focus on committees that are
preferred to unselected candidates by a large fraction of voters.
In particular, we argue that when $n$ voters select $k$ candidates, 
the winning committee should be preferred to each non-member by roughly $n-\frac{n}{k}$ 
voters. The resulting concept, which we call
{\em local stability}, can be seen as 
a translation of the notion of justified representation by
Aziz et al.~\cite{aziz:j:jr} from the world of approval-based elections to ranked-ballot elections. 
We also consider a stronger variant of this notion, which can be seen as an analogue
of extended justified representation from the work of Aziz et al.~\cite{aziz:j:jr}.

The goal of our work is to contrast the approach based on the ideas of
Gehrlein~\cite{Gehrlein1985199} and Ratliff~\cite{ratliff2003some}, which we call Gehrlein stability, 
with the approach based on Condorcet winning sets (i.e., local stability). By considering several restricted
domains (single-peaked, single-crossing, and a restriction implied by
the existence of political parties), we show that Gehrlein stable committees
are very well-suited for shortlisting (as already suggested by
Barber\'a and Coelho~\cite{bar-coe:j:non-controversial-k-names}),
whereas locally stable committees are better at providing proportional
representation.
From the point of view of the computational complexity, while in both
cases we show $\np$-hardness of testing the existence of a respective
``Condorcet committee'', we discover that a variant of the Gehrlein--Ratliff 
approach leads to a polynomial-time algorithm.

\section{Preliminaries}
\noindent
For every natural number $p$, we let $[p]$ denote the set $\{1, 2,
\ldots, p\}$.

An {\em election} is a pair $E = (C, V)$, where $C = \{c_1, \ldots,
c_m\}$ is a set of {\em candidates} and $V = (v_1, \ldots, v_n)$ is a
list of {\em voters}; we write $|V|$ to denote the number of voters in
$V$.  Each voter $v\in V$ is endowed with a linear {\em preference
  order} over $C$, denoted by $\succ_{v}$. For $\ell \in [m]$ we write
$\best_{\ell}(v)$ to denote $\ell$ candidates most preferred by
$v$. We write $\best(v)$ to denote the single most preferred candidate
of voter $v$, i.e., $\best_{1}(v) = \{\best(v)\}$, and for each $c\in
C\setminus \best(v)$ it holds that $\best(v) \succ_v c$. We write
$a\succ_v b \succ_v \dots$ to indicate that $v$ ranks $a$ first and
$b$ second, followed by all the other candidates in an arbitrary
order. Given two disjoint subsets of candidates $S, T\subseteq C$,
$S\cap T=\emptyset$, we write $S\succ_v T$ to indicate that $v$
prefers each candidate in $S$ to each candidate in $T$.

A {\em committee} is a subset of $C$. A {\em multiwinner voting rule}
$\calR$ takes an election $E=(C, V)$ and a positive integer $k$ with
$k\le |C|$ as its input, and outputs a non-empty collection of
size-$k$ committees. A multiwinner rule $\calR$ is said to be {\em
  resolute} if the set $\calR(E, k)$ is a singleton for each election
$E=(C, V)$ and each committee size~$k$. %

Given an election $E=(C, V)$ with $|V|=n$ and two candidates $c, d\in
C$, we say that $c$ {\em wins the pairwise election between $c$ and
  $d$} if more than $\frac{n}{2}$ voters in $V$ prefer $c$ to $d$; if exactly
$\frac{n}{2}$ voters in $V$ prefer $c$ to $d$, we say that the {\em pairwise
  election between $c$ and $d$ is tied}.  The {\em majority graph} of
an election $E = (C,V)$ is a directed graph $M(E)$ with vertex set $C$
and the following edge set:
\begin{align*}
\left\{(c, d)\in C^2 \mid \text{$c$ wins the pairwise
election between $c$ and $d$}\right\}.
\end{align*}
Observe that if the number of voters 
$n$ is odd, then $M(E)$ is a {\em tournament}, i.e., for each pair
of candidates $c, d\in C$ exactly one of their connecting edges, $(c, d)$ or $(d, c)$,
is present in $M(E)$. We will also consider
the {\em weak majority graph of $E$}, which we denote by $W(E)$:
this is the directed graph obtained from $M(E)$ by adding edges 
$(c, d)$ and $(d, c)$ for each pair of candidates $c, d$ such that
the pairwise election between $c$ and $d$ is tied.
A candidate $c$ is said to be a {\em Condorcet winner} of an election $E=(C, V)$
if the outdegree of $c$ in $M(E)$ is $|C|-1$; $c$ is said to be a {\em weak Condorcet winner}
of $E$ if the outdegree of $c$ in $W(E)$ is $|C|-1$.

\section{Gehrlein Stability and Local Stability}\label{sec:defs}
\noindent
Gehrlein~\cite{Gehrlein1985199} proposed a simple, natural extension
of the notion of a weak Condorcet winner to the case of multiwinner
elections, and a similar definition was subsequently introduced by
Ratliff~\cite{ratliff2003some}.  We recall and discuss Gehrlein's
definition, and then put forward a different approach to defining good
committees, which is inspired by the recent work on Condorcet winning
sets~\cite{journals/scw/ElkindLS15} and on justified representation in
approval-based committee elections~\cite{aziz:j:jr}.

\subsection{Gehrlein Stability}
\noindent
Gehrlein~\cite{Gehrlein1985199} and Ratliff~\cite{ratliff2003some}
base their approach on the following idea: a committee is unstable if
there exists a majority of voters who prefer a candidate that is not
currently in the committee to some current committee member.

\begin{definition}[Gehrlein~\cite{Gehrlein1985199}; Ratliff~\cite{ratliff2003some}]
  Consider an election $E = (C,V)$. A committee $S \subseteq C$ is
  \emph{weakly Gehrlein-stable} if for each committee member $c \in S$
  and each non-member $d \in C \setminus S$ it holds that $c$ wins or
  ties the pairwise election between $c$ and $d$.  Committee $S$ is
  \emph{strongly Gehrlein-stable} if for each $c\in S$ and each
  $d\not\in S$ the pairwise election between $c$ and $d$ is won by
  $c$. 
\end{definition}

By definition, each strongly Gehrlein-stable committee is also weakly
Gehrlein-stable, and the two notions are equivalent if the majority
graph $M(E)$ is a tournament.  Further, a strongly (respectively,
weakly) Gehrlein-stable committee of size one is simply a Condorcet
winner (respectively, a weak Condorcet winner) of a given election.
More generally, each member of a strongly (respectively, weakly)
Gehrlein-stable committee would be a Condorcet winner (respectively, a
weak Condorcet winner) should the other committee members be removed
from the election.  Note also that given a committee $S$, it is
straightforward to verify if it is strongly (respectively, weakly)
Gehrlein-stable: it suffices to check that there is no candidate in
$C\setminus S$ that ties or defeats (respectively, defeats) some
member of $S$ in their pairwise election.

Gehrlein stability has received some attention in the literature. In
particular, Ratliff~\cite{ratliff2003some},
Coelho~\cite{coelho2004understanding}, and, very recently,
Kamwa~\cite{Kamwa2015b}, proposed and analyzed a number of multiwinner
rules that satisfy weak Gehrlein stability, i.e., elect weakly
Gehrlein stable committees whenever they exist. These rules can be
seen as analogues of classic single-winner Condorcet-consistent rules,
such as Maximin or Copeland's rule (see, e.g., the survey by
Zwicker~\cite{ZwickerChapter} for their definitions).  Specifically,
each of these rules is based on a function that assigns non-negative
scores to committees in such a way that committees with score $0$ are
exactly the weakly Gehrlein-stable committees; it then outputs the
committees with the minimum score.

\smallskip

\noindent{\bf Gehrlein Stability and Majority Graphs\ } 
Gehrlein stability is closely related to a classic tournament solution
concept, namely, the top cycle (see, e.g., the survey by Brandt, Brill
and Harrenstein~\cite{BrandtEtAlChapter} for an overview of tournament
solution concepts).  Indeed, if the majority graph $M(E)$ is a
tournament, then every top cycle in $M(E)$ is a Gehrlein-stable
committee (recall that for tournaments weak Gehrlein stability is
equivalent to strong Gehrlein stability, so we use the term `Gehrlein
stability' to refer to both notions). In the presence of ties, the
relevant solution concepts are the Smith set and the Schwarz set: the
former corresponds to a weakly Gehrlein-stable committee and the
latter corresponds to a strongly Gehrlein-stable committee.

However, there is an important difference between Gehrlein committees
and each of these tournament solution concepts: When computing a
tournament solution, we aim to minimize the number of elements in the
winning set, whereas in the context of multiwinner elections our goal
is to find a weakly/strongly Gehrlein-stable committee of a given
size.  This difference has interesting algorithmic implications. While
it is easy to find a Smith set for a given tournament, in
Section~\ref{sec:complexity} we show that it is NP-hard to determine
if a given election admits a weakly Gehrlein-stable committee of a
given size.  On the other hand, we can extend the existing algorithm
for finding a Schwarz set to identify a strongly Gehrlein-stable
committee.
We defer most of our computational results until
Section~\ref{sec:complexity}, but we present the proof of this result
here because it implicitly provides a very useful characterization of
strongly Gehrlein-stable committees.

\begin{theorem}\label{thm:sg-easy}
Given an election $E=(C, V)$ and a positive integer $k$ with $k\le |C|$,
we can decide in polynomial time whether $E$ admits a strongly Gehrlein-stable 
committee of size $k$. Moreover, if such a committee exists, then it is unique.
\end{theorem}
\begin{proof}
Given an election $E=(C, V)$, we let ${\mathcal C} = \{C_1, \dots, C_r\}$
be the list of strongly connected components of $W(E)$; note that
a graph can be decomposed into strongly connected components in polynomial time.
Given two candidates $a, b\in C$, we write $a\to b$ if $W(E)$ contains a directed path from $a$ to $b$.

Consider two distinct sets $C_i, C_j\in {\mathcal C}$ and two candidates $a\in C_i$, $b\in C_j$.
Note that the pairwise election between $a$ and $b$ cannot be tied, since otherwise
$a$ and $b$ would be in the same set. Suppose without loss of generality
that $a$ beats $b$ in their pairwise election. Then for each $a'\in C_i$, $b'\in C_j$
we have $a'\to a$, $a\to b$, $b\to b'$ and hence by transitivity
$a'\to b'$. On the other hand, we cannot have $b'\to a'$, as this would
mean that $a'$ and $b'$ belong to the same connected component of $W(E)$.
Thus, we can define a total order on $\mathcal C$ as follows:
for $C_i, C_j\in {\mathcal C}$ we set $C_i < C_j$ if $i\neq j$ and $a\to b$
for each $a\in C_i$, $b\in C_j$. By the argument above, $<$ is indeed a total order on $\mathcal C$;
we can renumber the elements of $\mathcal C$ so that $C_1<\dots<C_r$.
Then for $a\in C_i$, $b\in C_j$ we have $a\to b$ if and only if $i\le j$.

Now, consider a strongly Gehrlein-stable committee $S$. Suppose that $a\to b$ and $b\in S$.
It is easy to see that $a\in S$; this follows by induction on the length of the 
shortest path from $a$ to $b$ in $W(E)$.
Hence, every strongly Gehrlein-stable committee is of the form $\bigcup_{i\le s} C_i$
for some $s\in [r]$. Thus, there is a strong Gehrlein committee of size $k$
if and only if $\sum_{i=1}^s|C_i|=k$ for some $s\in[r]$. This argument also shows
that a strongly Gehrlein-stable committee of a given size is unique.
\end{proof}

We have argued that for tournaments the notions of weak Gehrlein stability
and strong Gehrlein stability coincide. We obtain the following corollary.

\begin{corollary}
Consider an election $E=(C, V)$. If $W(E)$ is a tournament, we can decide
in polynomial time whether $E$ has a weakly Gehrlein-stable committee of a given size.
Moreover, if such a committee exists, it is unique.
\end{corollary}

\smallskip

\noindent{\bf Gehrlein Stability and Enlargement Consistency\ }
Interestingly, Barber\'a and Coelho~\cite{bar-coe:j:non-controversial-k-names} 
have shown that weak Gehrlein stability is incompatible with \emph{enlargement
  consistency}: For every resolute multiwinner rule $\calR$ that
elects a weakly Gehrlein-stable committee whenever such a committee exists, 
there exists an election $E$ and
committee size $k$ such that the only committee in $\calR(E,k)$ is not
a subset of the only committee in $\calR(E,k+1)$.\footnote{Enlargement
  consistency is defined for resolute rules only. An analogue of this notion for
  non-resolute rules was introduced by Elkind et
  al.~\cite{elk-fal-sko-sli:c:multiwinner-rules} under the name of
  \emph{committee monotonicity}.} 
While this result means that such rules are not well-suited  
for shortlisting tasks~\cite{bar-coe:j:non-controversial-k-names,elk-fal-sko-sli:c:multiwinner-rules},
it only holds for weak Gehrlein stability and not for strong Gehrlein stability.

Indeed, let us consider the following multiwinner variant of the Copeland
  rule (it is very similar to the NED rule of
  Coelho~\cite{coelho2004understanding} and we will call it
  \emph{strong-NED}). Given an election $E$, the score of a candidate
  is its outdegree in $M(E)$. Strong-NED chooses the committee of $k$
  candidates with the highest scores (to match the framework of
  Barber\'a and Coelho~\cite{bar-coe:j:non-controversial-k-names}, the
  rule should be resolute and so we break ties lexicographically). By
  its very definition, strong-NED satisfies enlargement
  consistency. Further, if there is a committee $W$ that is strongly
  Gehrlein-stable, then strong-NED chooses this committee (if there
  are $m$ candidates in total, then the outdegree of each candidate
  from $W$ is at least $m-|W|$, whereas the outdegree of each
  candidate outside of $W$ is at most $m-|W|-1$; Barber\'a and
  Coelho~\cite{bar-coe:j:non-controversial-k-names} also gave this
  argument, but assuming an odd number of voters).

\subsection{Local Stability}
\noindent
An important feature of Gehrlein stability is that it is strongly
driven by the majority opinions.  Suppose, for instance, that a group
of 1000 voters is to elect 10 representatives from the set $\{c_1,
\dots, c_{20}\}$, and the society is strongly polarized: 501 voters
rank the candidates as $c_1\succ\dots\succ c_{20}$, whereas the
remaining 499 voters rank the candidates as $c_{20}\succ\dots\succ
c_1$. Then the unique Gehrlein-stable committee of size $10$ consists
of candidates $c_1, \dots, c_{10}$, and the preferences of 499 voters
are effectively ignored.  While this is appropriate in some settings,
in other cases we may want to ensure that candidates who are
well-liked by significant minorities of voters are also elected.

Aziz et al.~\cite{aziz:j:jr} formalize this idea in the context of
approval voting, where each voter submits a set of candidates that she
approves of (rather than a ranked ballot). Specifically, they say that
committee $S$, $|S|=k$, provides {\em justified representation} in an
election $(C, V)$ with $|V|=n$, where each voter $i$ is associated
with an approval ballot $A_i\subseteq C$, if there is no group of
voters $V'\subseteq V$ with $|V'|\ge \lceil\frac{n}{k}\rceil$ such
that $A_i\cap S=\emptyset$ for each $i\in V'$, yet there exists a
candidate $c\in C\setminus S$ approved by all voters in
$V'$. Informally speaking, this definition requires that each
`cohesive' group of voters of size at least $q=
\lceil\frac{n}{k}\rceil$ is represented in the committee.  The choice
of threshold $q=\lceil\frac{n}{k}\rceil$ (known as the {\em Hare
  quota}) is natural in the context of approval voting: it ensures
that, when the electorate is composed of $k$ equal-sized groups of
voters, with sets of candidates approved by each group being pairwise
disjoint, each group is allocated a representative.

Extending this idea to ordinal ballots and to an arbitrary threshold
$q$, we obtain the following definition.

\begin{definition}\label{def:ls}
  Consider an election $E = (C,V)$ with $|V|=n$ and a positive value $q\in{\mathbb Q}$. 
  A committee $S$ \emph{violates local stability for quota $q$} 
  if there exists a group $V^*\subseteq V$ with $|V^*|\ge q$ 
  and a candidate $c \in C \setminus S$ such that each voter from
  $V^*$ prefers $c$ to each member of $S$; otherwise, $S$ {\em 
  provides local stability for quota $q$}.
\end{definition}

Note that, while in the context of approval voting the notion of group cohesiveness
can be defined in absolute terms (a group is considered cohesive if there is a candidate 
approved by all group members), for ranked ballots a cohesive group is defined relative
to a given committee (a group is cohesive with respect to $S$ if all its members prefer
some candidate to $S$). Another important difference between the two settings is that, 
while a committee that provides justified representation is guaranteed to exist
and can be found in polynomial time~\cite{aziz:j:jr}, a committee that provides local stability
may fail to exist, even if we use the same value of the quota, i.e., $q=\lceil\frac{n}{k}\rceil$.

\begin{example}\label{ex:nols} {\em Fix an integer $d\ge 2$; let
    $X=\{x_1, \dots, x_d\}$, $Y=\{y_1, \dots, y_d\}$, $Z=\{z_1, \dots,
    z_d\}$, and set $C=\{a, b\}\cup X\cup Y\cup Z$.  There are $4$
    voters with preferences $a\succ b\succ\cdots$ and $4$ voters with
    preferences $b\succ a\succ\cdots$.  Also, for each $i\in[d]$,
    there are two voters with preferences $x_i\succ y_i\succ
    z_i\succ\cdots$, two voters with preferences $y_i\succ z_i\succ
    x_i\succ\cdots$, and two voters with preferences $z_i\succ
    x_i\succ y_i\succ\cdots$. Altogether, we have $n=6d+8$ voters.
    Set $k=2d+1$; then for $d\ge 2$ we obtain $\lceil\frac{n}{k}\rceil
    = 4$.  We will now argue that this election admits no locally
    stable committee of size $k$ for quota
    $q=\lceil\frac{n}{k}\rceil$. Suppose for the sake of contradiction
    that $S$ is a locally stable committee of size $k$ for this value
    of the quota.  Note first that for each $i\in[d]$ we have $|\{x_i,
    y_i, z_i\}\cap S|\ge 2$.  Indeed, suppose that this is not the
    case for some $i\in[d]$.  By symmetry, we can assume without loss
    of generality that $y_i, z_i\not\in S$. However, then there are
    $4$ voters who prefer $z_i$ to every member of the committee, a
    contradiction with local stability. Thus, $S$ contains at least
    $2d$ candidates in $X\cup Y \cup Z$ and hence $|S\cap\{a, b\}|\le
    1$.  Thus at least one of $a$ or $b$ does not belong to the
    committee and either the four $a \succ b \succ \cdots$ voters or
    the four $b \succ a \succ \cdots$ voter witness that $S$ is not
    locally stable.  or the second four voters
}
\end{example}

In Section~\ref{sec:complexity}, we will use the idea from
Example~\ref{ex:nols} to argue that it is NP-hard to decide whether a
given election admits a locally stable committee.

Definition~\ref{def:ls} does not specify a value of the quota
$q$. Intuitively, the considerations that should determine the choice
of quota are the same as for Single Transferable Vote (STV), and one
can choose any of the quotas that are used for STV (see, e.g., the
survey by Tideman~\cite{tideman95}).  In particular, for $k=1$ and the
Hare quota $q=\lceil\frac{n}{k}\rceil$ we obtain Pareto optimality: a
committee $\{a\}$ of size $k=1$ is locally stable for quota
$\lceil\frac{n}{k}\rceil$ if there is no other candidate $c$ such that
all voters prefer $c$ to $a$. For $k=1$ and
$q=\lceil\frac{n}{k+1}\rceil$ (the {\em Hagenbach-Bischoff
  quota}), %
locally stable committees for quota $q$ are those whose unique element
is a weak Condorcet winner; for $k=1$ and
$q=\lfloor\frac{n}{k+1}\rfloor+1$ (the {\em Droop quota}), a locally
stable committee for quota $q$ has the Condorcet winner as its only
member.

\begin{remark}
{\em
For $k=2$ and $q=\left\lceil\frac{n}{k}\right\rceil$, locally stable committees are closely
related to Condorcet winning sets, as defined by Elkind et
al.~\cite{journals/scw/ElkindLS15}, and, more generally, locally
stable committees are related to $\theta$-winning
sets~\cite{journals/scw/ElkindLS15}.  Elkind et
al.~\cite{journals/scw/ElkindLS15} say that a set of candidates $S$ is
a {\em $\theta$-winning set} in an election $(C, V)$ with $|V|=n$ if
for each candidate $c\in C\setminus S$ there are more than $\theta n$
voters who prefer some member of $S$ to $c$; a $\frac{1}{2}$-winning
set is called a {\em Condorcet winning set}.  Importantly, unlike
locally stable committees, $\theta$-winning sets are defined in terms
of strict inequalities. If we replace `more than $\theta n$' with `at
least $\theta n$' in the definition of Elkind et
al.~\cite{journals/scw/ElkindLS15}, we obtain the definition of local
stability for quota $q=(1-\theta)n$.  
Elkind et al.~\cite{journals/scw/ElkindLS15} define a voting rule that
for a given election $E$ and committee size $k$ outputs a size-$k$
$\theta$-winning set for the smallest possible $\theta$. This rule, by
definition, outputs locally stable committees whenever they exist.
We remark that the $15$-voter, $15$-candidate election
described by Elkind et al.~\cite{journals/scw/ElkindLS15} is an
example of an election with no locally stable committee for
$q=\left\lceil\frac{n}{k}\right\rceil$ and $k=2$, thus complementing Example~\ref{ex:nols}
(which works for odd $k\ge 5$).
}
\end{remark}

For concreteness, from now on we fix the quota to be
$q=\lfloor\frac{n}{k+1}\rfloor+1$ (the Droop quota), and use the
expression `locally stable committee' to refer to locally stable
committees for this value of the quota. However, some of our results
extend to other values of $q$ as well.

\smallskip

\noindent{\bf Full Local Stability\ }
Aziz et al.~\cite{aziz:j:jr} also proposed the notion of
\emph{extended justified representation}, which deals with larger
groups of voters that, intuitively, are entitled to more than a single
representative. To apply their idea to ranked ballots, we need to
explain how voters evaluate possible deviations.  We require a new
committee to be a Pareto improvement over the old one: given a committee $S$
and a size-$\ell$ set of candidates $T$, we say that a voter $v$
prefers $T$ to $S$ if there is a bijection $\mu: T \to \best_{\ell}(S)$ such that for each $c\in T$ voter $v$ weakly prefers
$c$ to $\mu(c)$ and for some $c\in T$ voter $v$ strictly prefers
$c$ to $\mu(c)$.  We now present our analogue of extended justified
representation for ranked ballots, which we call full local stability.

\begin{definition}\label{def:fls}
  Consider an election $(C,V)$ with $|V|=n$.  We say that a committee
  $S$, $|S|=k$, \emph{violates $\ell$-local stability} for
  $\ell\in[k]$ if there exists a group of voters $V^*\subseteq V$ with
  $|V^*|\ge \left\lfloor\frac{\ell\cdot n}{k+1}\right\rfloor+1$ and a
  set of $\ell$ candidates $T$, such that 
  each voter $v\in V^*$ prefers $T$ to $S$; otherwise, $S$ provides
  $\ell$-local stability.  A committee $S$ with $|S|=k$ provides
  \emph{full local stability} if it provides $\ell$-local stability
  for all $\ell\in[k]$.
\end{definition}

By construction, 1-local stability is simply local stability, and hence 
every committee that provides full local stability also provides local stability.
Local stability and full local stability 
extend from committees to voting rules in a natural way.

\begin{definition}
  A multiwinner voting rule $\calR$ {\em satisfies (full) local
    stability} if for every election $E=(C, V)$ and every target
  committee size $k$ such that in $E$ the set of size-$k$ committees
  that provide (full) local stability is not empty, it holds that
  every committee in $\calR(E)$ provides (full) local stability.
\end{definition}

Weakly/strongly Gehrlein-stable rules can be defined in a similar manner.

\smallskip

\noindent{\bf Solid Coalitions and Dummett's Proportionality\ }
Let us examine the relation between (full) local stability, and the
solid coalitions property and Dummett's proportionality. Both these
notions were used by Elkind et
al.~\cite{elk-fal-sko-sli:c:multiwinner-rules} as indicators of voting
rules' ability to find committees that represent voters proportionally
(however, we give a slightly different definition than they give; see
explanation below).

\begin{definition}\label{def:solidCoalitions}
  Consider an election $(C,V)$ with $|V|=n$.  We say that a committee
  $S$, $|S|=k$, {\em violates the solid coalitions property} if there
  exists a candidate $c \notin S$ who is ranked first by some $\lceil
  \frac{n}{k} \rceil$ voters.  We say that a committee $S$, $|S|=k$,
  {\em violates Dummett's proportionality} if there exists a set of
  $\ell$ candidates $Q$ with  $Q \setminus S \neq \emptyset$ and a set of $\lceil
  \frac{\ell n}{k} \rceil$ voters $V$ such that for each voter $v\in V$
  it holds that $\best_{\ell}(v)=Q$.
\end{definition}

\begin{proposition}\label{prop:relationToSolidCoalitions}
A locally stable committee satisfies the solid coalitions property;
a fully locally stable committee satisfies Dummett's proportionality.
\end{proposition}
\begin{proof}
  We present the proof for local stability; for full local stability
  the same argument can be used.  Consider an election $E$, a target
  committee size $k$, and a committee $S$ such that some $\lceil
  \frac{n}{k} \rceil$ voters rank a candidate $c \in C \setminus S$
  first.  Since $\frac{n}{k} > \frac{n}{k+1}$, also $\lceil \frac{n}{k} \rceil
  \geq \lfloor\frac{n}{k+1}\rfloor+1$, and so the same group of voters
  witnesses that $S$ violates local stability.
\end{proof}

The solid coalitions property and Dummett's proportionality are
usually defined as properties of multiwinner rules. In contrast,
Definition~\ref{def:solidCoalitions} treats them as properties of
coalitions, which is essential for establishing a relation such as the
one given in Proposition~\ref{prop:relationToSolidCoalitions}. Indeed,
local stability as the property of a rule puts no restrictions on the
output of the rule for profiles for which there exists no locally
stable committees and, in particular, for such profiles local
stability does not guarantee the solid coalitions property.

\section{Three Restricted Domains}
\noindent
In Section~\ref{sec:defs} we have argued that Gehrlein stability is a
majoritarian notion, whereas local stability is directed towards
proportional representation. Now we reinforce this intuition by
describing the structure of Gehrlein-stable and locally stable committees for three well-studied
restricted preference domains.  Namely, we consider
single-crossing elections, single-peaked elections, and elections
where the voters have preferences over parties (modeled as large sets
of `similar' candidates).

The following observation will be useful in our analysis.  Consider an
election $E=(C, V)$ for which $M(E)$ is a transitive tournament, i.e.,
if $(a, b)$ and $(b, c)$ are edges of $M(E)$ then $(a, c)$ is also an
edge of $M(E)$.  In such a case, the set of ordered pairs $(a, b)$
such that $(a, b)\in M(E)$ is a linear order on $C$ and we refer to it
as the \emph{majority preference order}. Given a positive integer $k$,
we let the \emph{centrist committee} $S_\mathrm{center}$ consist of
the top $k$ candidates in the majority preference order.
Theorem~\ref{thm:sg-easy} implies the following simple observation.

\begin{proposition}\label{prop:centrist}
  If $M(E)$ is transitive then for each committee size $k$,
  $S_\mathrm{center}$ is strongly Gehrlein-stable.
\end{proposition}

It is well known that if the number of voters is odd and the election
is either single-peaked or single-crossing (see definitions below),
then $M(E)$ is a transitive tournament. Thus
Proposition~\ref{prop:centrist} is very useful in such settings.

\subsection{Single-Crossing Preferences}\label{sec:sc}
\noindent
The notion of single-crossing preferences was proposed by
Mirrlees~\cite{mir:j:single-crossing} and Roberts~\cite{rob:j:tax}.
Informally speaking, an election is single-crossing if (the voters can
be ordered in such a way that) as we move from the first voter to the
last one, the relative order within each pair of candidates changes at
most once.  For a review of examples where single-crossing preferences
can arise, we refer the reader to the work of Saporiti and
Tohm\'e~\cite{sap-toh:j:single-crossing-strategy-proofness}.

\begin{definition}\label{def:sc}
  An election $(C,V)$ with $V=(v_1, \dots, v_n)$ is
  {\em single-crossing}%
  \footnote{Our definition of single-crossing elections
  assumes that the order of voters is fixed. More commonly, an election is defined to be single-crossing if voters
  can be permuted so that the condition formulated in Definition~\ref{def:sc} holds. For our purposes, this distinction
  is not important, and the approach we chose makes the presentation more compact.} 
  if for each pair of candidates $a, b \in C$ such that
  $v_1$ prefers $a$ over $b$ we have  $\{i \mid a \succ_{v_i} b\} = [t]$
  for some $t\in[n]$.
\end{definition}

Single-crossing elections have many desirable properties. In the
context of our work, the most important one is that if $E=(C, V)$ is a
single-crossing election with an odd number of voters, then $M(E)$ is
a transitive tournament. Moreover if $|V|=2n'+1$, the majority
preference order coincides with the preferences of the $(n'+1)$-st
voter~\cite{rothstein:j:representative-voters}.  By
Proposition~\ref{prop:centrist}, this means that for single-crossing
elections with an odd number of voters the centrist committee exists,
is strongly Gehrlein-stable, and consists of the top $k$ candidates in
the preference ranking of the median voter, which justifies the term
\emph{centrist committee}.

\begin{proposition}\label{prop:single_crossing}
  For a single-crossing election with an odd number of voters, the
  centrist committee is strongly Gehrlein-stable.
\end{proposition}

Locally stable committees turn out to be very different.  Let $E = (C,
V)$ be a single-crossing election with $|V|=n$, and let $k$ be the
target committee size; then the Droop quota for $E$ is $q =
\left\lfloor\frac{n}{k+1}\right\rfloor+1$.  We say that a size-$k$
committee $S$ is \emph{single-crossing uniform} for $E$ if for each
$\ell \in [k]$ it contains the candidate ranked first by voter
$v_{\ell\cdot q}$.  Note that a single-crossing uniform committee need
not be unique: e.g., if all the voters rank the same candidate first,
then every committee containing this candidate is single-crossing
uniform.

\begin{example}\label{ex:sc}
{\em
  Figure~\ref{fig:proportional_and_centrist_sc}
  shows a single-crossing election with $15$ voters
  over the candidate
  set $C = \{a,b,c,d,e,f,g,h,i,j,k\}$.
  The first voter ranks the candidates in the alphabetic order, 
  and the last voter ranks them in the reverse alphabetic order.  
  For readability, we list the top four-ranked candidates only. 
  For the target committee size $4$, the centrist
  committee (marked with a rectangle) is $\{c,d,e,f\}$,
  and the unique single-crossing uniform committee is
  $\{b,d,g,i\}$ (marked with dashed ellipses).
  If we reorder the voters from $v_{15}$
  to $v_1$, the unique single-crossing uniform committee is $\{c,d,h,j\}$.
}
\end{example}

We will now argue that 
single-crossing uniform committees are locally stable.

\newcommand{\vote}[5]{
  \node at ( 0.5+#1*0.5 ,1 ) {$v_{#1}$};
  \node at ( 0.5+#1*0.5 ,0.5 ) {$\strut #2$};
  \node at ( 0.5+#1*0.5 ,0.1 ) {$\strut #3$};
  \node at ( 0.5+#1*0.5 ,-0.3 ) {$\strut #4$};
  \node at ( 0.5+#1*0.5 ,-0.7 ) {$\strut #5$};
  \node at ( 0.5+#1*0.5 ,-1.1 ) {$\strut \vdots$};
}
\newcommand{\mvote}[7]{}
\begin{figure}[t]
\centering
\begin{tikzpicture}
  \draw (0.75,0.75) -- (8.25,0.75);

  \draw[very thick] (4.3,0.73) rectangle (4.7,-0.9);
  \draw[ultra thick, dashed] (2,0.75) ellipse (2.5mm and 5mm);
  \draw[ultra thick, dashed] (3.5,0.75) ellipse (2.5mm and 5mm);
  \draw[ultra thick, dashed] (5,0.75) ellipse (2.5mm and 5mm);
  \draw[ultra thick, dashed] (6.5,0.75) ellipse (2.5mm and 5mm);

  \vote{1}  a b c d \mvote e f g h i j k
  \vote{2}  a b c d \mvote e f g h i j k
  \vote{3}  b c a d \mvote e f g h i j k
  \vote{4}  c b a d \mvote e f g h i j k
  \vote{5}  c b d a \mvote e f g h i j k
  \vote{6}  d c b a \mvote e f g h i j k
  \vote{7}  d c e b \mvote a f g h i j k
  \vote{8}  f d e c \mvote b a g h i j k
  \vote{9}  g h f d \mvote e c b a i j k
  \vote{10} h g f d \mvote e c b a i j k
  \vote{11} h g f i \mvote e d c b a j k
  \vote{12} i h g f \mvote e d c b a j k
  \vote{13} j i h g \mvote f e d c b a k
  \vote{14} j i k h \mvote g f e d c b a 
  \vote{15} k j i h \mvote g f e d c b a 
\end{tikzpicture}
\caption{A single-crossing election (Example~\ref{ex:sc}).}

\label{fig:proportional_and_centrist_sc}
\end{figure}

\begin{proposition}\label{prop:single-crossing-local-stability}
  For every single-crossing election $E=(C, V)$ 
  and for every $k\in[|C|]$ it holds that every size-$k$ single-crossing
  uniform committee for $E$ is locally stable.
\end{proposition}
\begin{proof}
  Fix a single-crossing election $E=(C, V)$ with $|V|=n$
  and a target committee size $k$; set 
  $q=\left\lfloor\frac{n}{k+1}\right\rfloor+1$.
  Consider a committee $S$, $|S|=k$, that is single-crossing uniform 
  with respect to $E$. We will show that $S$ is locally stable.

  Consider an arbitrary candidate $c\not\in S$. Suppose first that some voter $v_i$
  with $i<q$ ranks $c$ above all candidates in $S$. Let $a=\best(v_q)$.
  As $a\in S$  and $E$ is single-crossing, each voter $v_j$ with $j\ge q$
  prefers $a$ to $c$. Thus, there are at most $q-1$ voters who prefer $c$ 
  to each member of $S$. 
   
  Now, suppose that some voter $v_i$ with $\ell q < i < (\ell+1)q$ for some $\ell\in[k-1]$
  ranks $c$ above all candidates in $S$; let 
  $a=\best(v_{\ell\cdot q})$, $b=\best(v_{(\ell+1)\cdot q})$. 
  By construction we have $a, b\in S$ and by the single-crossing property $a\neq b$
  (if $a=b$, then $a$ and $c$ would cross more than once).
  Also, by the single-crossing property all voters $v_j$ with $j\le \ell q$
  rank $a$ above $c$ and all voters $v_{j'}$ with $j'\ge (\ell+1)q$ rank $b$ above $c$.
  Thus, there are at most $q-1$ voters who prefer $c$ to each member of $S$.

  Finally, suppose that some voter $v_i$ with $i>kq$ ranks $c$ above all members of $S$;
  let $a=\best(v_{k\cdot q})$. We have $a\in S$ and
  by the single-crossing property all voters $v_j$ with $j\le kq$ rank $a$ above $c$.
  Thus, there are at most $n-kq$ voters who may prefer $c$ to $a$, 
  and $q> \frac{n}{k+1}$ implies $n - qk < n - \frac{nk}{k+1} = \frac{n}{k+1} < q$.
  In each case, the number of voters who may prefer $c$ to all members of $S$
  is strictly less than~$q$.
\end{proof}

The following example shows that a single-crossing uniform committee can violate Gehrlein stability
and, similarly, that the centrist committee can violate local stability.

\begin{example}\label{ex:stability_violation}
{\em
Let $C=\{a, b, c\}$. Consider the single-crossing election
where three voters rank the candidates as $a\succ b\succ c$
and four voters rank the candidates as $c\succ b\succ a$.
Let $k = 2$. We have $\left\lfloor\frac{n}{k+1}\right\rfloor+1 = 3$.
The committee $\{a, c\}$ is single-crossing uniform for this election, 
yet four voters out of seven prefer $b$ to $a$.
The committee $\{b, c\}$ is centrist, yet it is not locally stable 
since there are $q=3$ voters who prefer $a$ to 
both $b$ and $c$.
}
\end{example}

\subsection{Single-Peaked Preferences}\label{sec:sp}
\noindent
The class of single-peaked preferences, 
first introduced by Black~\cite{bla:j:rationale-of-group-decision-making}, is perhaps the most 
extensively studied restricted preference domain.

\begin{definition}
  Let $\lhd$ be an order over $C$. We say that an election $E=(C, V)$ 
  is {\em single-peaked with respect to $\lhd$} if for each voter $v\in V$ 
  and for each pair of candidates $a, b \in C$
  such that $\best(v) \lhd a \lhd b$ or $b \lhd a \lhd \best(v)$ it holds that $a \succ_v b$. 
  We will refer to $\lhd$ as a \emph{societal axis} for $E$.
\end{definition}

Just as in single-crossing elections, in single-peaked elections with an odd number of voters
the majority preference order is transitive and hence the centrist committee is well-defined
and strongly Gehrlein-stable.

\begin{proposition}\label{prop:single_peaked}
  For a single-peaked election with an odd number of voters, the
  centrist committee is strongly Gehrlein-stable.
\end{proposition}

Moreover, we can define an analogue of a single-crossing uniform committee for single-peaked elections.
To this end, given an election $E = (C, V)$ that is single-peaked with respect to the societal axis
$\lhd$, we reorder the voters so that for the new order $V'=(v_1', \dots, v'_n)$
it holds that $\best(v'_i)\lhd \best(v'_j)$ implies $i<j$; we say
that an order of voters $V'$ that has this property is {\em $\lhd$-compatible}.
We can now use the same construction as in Section~\ref{sec:sc}. 
Specifically, given an election $E=(C, V)$ with $|V|=n$ that is single-peaked with respect to $\lhd$
and a target committee size $k$, we set $q=\left\lfloor\frac{n}{k+1}\right\rfloor+1$, 
and say that a committee $S$ is {\em single-peaked uniform} for $E$ if
for some $\lhd$-compatible order of voters $V'$  
we have $\best(v'_{\ell\cdot q})\in S$ for each $\ell\in[k]$.

\begin{proposition}\label{prop:single-peaked-local-stability}
For every single-peaked election $E=(C, V)$ 
  and for every $k\in[|C|]$ it holds that every size-$k$ single-peaked
  uniform committee for $E$ is locally stable.
\end{proposition}

\begin{proof}
  Fix an election $E=(C, V)$ with $|C|=m$, $|V|=n$ that is single-peaked with respect to $\lhd$
  and a target committee size $k$. Assume without loss of generality that 
  $\lhd$ orders the candidates as $c_1\lhd\dots\lhd c_m$ and that $V$ is $\lhd$-compatible.
  Consider a single-peaked uniform committee $S$ of size $k$.
  Recall that $q=\left\lfloor\frac{n}{k+1}\right\rfloor+1$.
  Let $c_\ell=\best(v_q)$, $c_r=\best(v_{k\cdot q})$.
  Consider a candidate $c_j\not\in S$.
  
  Suppose first that $j<\ell$. Then for each voter $v_i$ with $i\ge q$
  the candidate $\best(v_i)$ is either $c_\ell$ or some candidate to the right 
  of $c_\ell$. Thus, all such voters prefer $c_\ell$ to $c_j$, and hence
  there can be at most $q-1$ voters who prefer $c_j$ to each member of $S$.
  By a similar argument, if $j>r$, there are at most $n-kq < q$ voters
  who prefer $c_j$ to each member of $S$.

  It remains to consider the case $\ell<j<r$. Let $\ell'=\max\{t\mid t<j, c_t\in S\}$,
  $r'=\min\{t\mid t>j, c_t\in S\}$. Set $i=\max\{i: \best(v_{i\cdot q})=c_{\ell'}\}$;
  then the most preferred candidate of voter $v_{(i+1)\cdot q}$ is $c_{r'}$. 
  Since the voters' preferences are single-peaked with respect to $\lhd$,
  $v_{i\cdot q}$ and all voters that precede her in $V$ prefer $c_{\ell'}$
  to $c_j$, and $v_{(i+1)\cdot q}$ and all voters that appear after her in $V$ prefer $c_{r'}$
  to $c_j$. Thus, only the voters in the set $V' = \{v_{i\cdot q+1}, \dots, v_{i\cdot q+q-1}\}$
  may prefer $c_j$ to all voters in $S$, and $|V'|\le q-1$.

  Thus, for any choice of $c_j\not\in S$ fewer than $q$ voters prefer $c_j$ to all members
  of $S$, and hence $S$ is locally stable. 
\end{proof}

The proof of Proposition~\ref{prop:single-peaked-local-stability} is very similar 
to the proof of Proposition~\ref{prop:single-crossing-local-stability}; we omit it due to space constraints.

Observe that the election from Example~\ref{ex:stability_violation} is single-peaked
and committee $\{a, c\}$ is single-peaked uniform for that election. 
This shows that for single-peaked elections a single-peaked uniform committee 
can violate Gehrlein stability
and the centrist committee can violate local stability.

\subsection{Party-List Elections}\label{sec:pl}
\noindent
When candidates are affiliated with political parties, it is not unusual for the voters'
preferences to be driven by party affiliations: a voter who associates herself with a political
party, ranks the candidates who belong to that party above all other candidates
(but may rank candidates that belong to other parties arbitrarily). In the presence
of a strong party discipline, we may additionally assume that all supporters of a given
party rank candidates from that party in the same way. We will call elections
with this property {\em party-list elections}.

\begin{definition}
An election $E=(C, V)$ is said to be a {\em party-list election 
for a target committee size $k$} if we can partition the set of candidates 
$C$ into pairwise disjoint sets $C_1,\dots, C_p$ and the set of voters $V$ into pairwise disjoint groups $V_1, \dots, V_p$ 
so that
\begin{inparaenum}[(i)]
  \item $|C_i| \geq k$ for each $i \in [p]$,
  \item each voter from $V_i$ prefers each candidate in $C_i$ to each candidate in $C\setminus C_i$, 
  \item for each $i \in [p]$ 
        all voters in $V_i$ order the candidates in $C_i$ in the same way.
\end{inparaenum}
\end{definition}

Party-list elections are helpful for understanding the difference
between local stability and full local stability.  Indeed, when all
voters have the same preferences over candidates, local stability only
ensures that a committee contains the unanimously most preferred
candidate. In particular, the committee that consists of the single
most preferred candidate and the $k-1$ least preferred candidates is
locally stable. On the other hand, full local stability imposes
additional constraints.  For example, when preferences are unanimous,
only the committee that consists of the $k$ most preferred candidates
satisfies full local stability.  Generalizing this observation, we
will now show that in party-list elections a fully locally stable
committee selects representatives from each set $C_i$ in proportion to
the number of voters in $V_i$.

\begin{theorem}\label{thm:partylist}
Let $E=(C, V)$ be a party-list election for a target committee size $k$,
and let $(C_1, \dots, C_p)$ and $(V_1, \dots, V_p)$ be the respective partitions of $C$ and $V$.
Then for each $i\in[p]$ 
every committee $S$ of size $k$ that provides full local stability for $E$
contains all candidates ranked in top $\left\lfloor k \cdot \frac{|V_i|}{n}\right\rfloor$ 
positions by the voters in $V_i$.
\end{theorem}
\begin{proof}
Consider a committee $S$, $|S|=k$, that provides full local stability for $E$.
Fix $i\in[p]$, let $\ell=\left\lfloor k \frac{|V_i|}{n}\right\rfloor$,
and let $C'_i$ be the set of candidates ranked in top $\ell$ positions by each voter in $V_i$.
Let $c$ be some candidate in $C'_i$. If $c\not\in S$, then voters in $V_i$ prefer $C'_i$ to $S$.  
As 
\begin{align*}
|V_i| \geq \left\lfloor \frac{|V_i| \cdot k}{k+1} \right\rfloor+1 = \left\lfloor k\frac{|V_i|}{n} \cdot \frac{n}{k+1} \right\rfloor+1 \geq \left\lfloor \ell\frac{n}{k+1}\right\rfloor+1 \text{,}
\end{align*}
this would mean that $S$ violates $\ell$-local stability for $E$,
a contradiction.
\end{proof}

\begin{example}\label{example:partylist}
{\em
Consider an election $E=(C, V)$ with $C=X\cup Y\cup Z$, 
$X=\{x_1, \dots, x_4\}$, $Y=\{y_1, \dots, y_4\}$, $Z=\{z_1, \dots, z_4\}$
and $|V|=16$, where $8$ voters rank the candidates as 
$x_1\succ x_2 \succ x_3 \succ x_4\succ \dots$, 
$4$ voters rank the candidates as 
$y_1\succ y_2 \succ y_3 \succ y_4\succ \dots$ and 
$4$ voters rank the candidates as 
$z_1\succ z_2 \succ z_3 \succ z_4\succ \dots$.
Let $k=4$. Clearly, $E$ is a party-list election.
To provide full local stability, 
a committee has to contain the top two candidates from $X$, 
the top candidate from $Y$ and the top candidate from $Z$.
Thus, $\{x_1, x_2, y_1, z_1\}$ is the unique fully locally stable committee.
}
\end{example}

On the other hand, observe that in an election where two parties have
equal support, i.e., when $C$ and $V$ are partitioned into $C_1, C_2$
and $V_1, V_2$, respectively, and $|V_1|=|V_2|$, every committee $S$
that contains the top candidate in $C_1$ (according to voters in
$V_1$) and the top candidate in $C_2$ (according to voters in $V_2$),
provides local stability.  Thus, local stability can capture the idea
of diversity to some extent, but not of fully proportional
representation.

Finally, note that Gehrlein stability does not offer any guarantees in
the party-list framework: If a party is supported by more than half of
the voters, then the top $k$ candidates of this party form the unique
strongly Gehrlein-stable committee; if a party is supported by
fewer than half of the voters then it is possible that none of its
candidates is in a weakly Gehrlein-stable committee.

\section{Computational Complexity}\label{sec:complexity}
\noindent
We will now argue that finding stable committees can be computationally challenging,
both for weak Gehrlein stability and for local stability (recall that, in contrast,
for strong Gehrlein stability Theorem~\ref{thm:sg-easy} provides a polynomial-time algorithm).
Full local stability appears to be even more demanding:
we provide evidence that even checking whether a given committee 
is fully locally stable is hard as well.

\begin{theorem}\label{thm:wg-hard}
Given an election $E = (C, V)$ and a target committee size $k$ with $k\le |C|$,
it is {\em NP}-complete to decide if there exists a weakly Gehrlein-stable committee
of size $k$ for $E$.
\end{theorem}
\begin{proof}
  It is immediate that this problem is in NP: given an election $E =
  (C, V)$, a target committee size $k$, and a committee $S$ with
  $|S|=k$, we can check that $S$ has no incoming edges in $M(E)$.

  To show hardness, we provide a reduction from \textsc{Partially
    Ordered Knapsack}.  An instance of this problem is given by a list
  of $r$ ordered pairs of positive integers ${\mathcal L} =
  \left((s_1, w_1), \dots, (s_r, w_r)\right)$, a capacity bound $b$, a
  target weight $t$, and a directed acyclic graph $\Gamma=([r],
  A)$. It is a `yes'-instance if there is a subset of indices
  $I\subseteq [r]$ such that $\sum_{i\in I}s_i\le b$, $\sum_{i\in
    I}w_i\ge t$ and for each directed edge $(i, j)\in A$ it holds that
  $j\in I$ implies $i\in I$.  This problem is strongly NP-complete;
  indeed, it remains NP-hard if $s_i=w_i$ and $w_i\le r$ for all
  $i\in[r]$~\cite{joh-nie:j:poknapsack}.  Note that if $s_i=w_i$ for
  all $i\in[r]$, we can assume that $b=t$, since otherwise we
  obviously have a `no'-instance.

  Given an instance $\langle{\mathcal L}, b, t, \Gamma \rangle$ of
  {\sc Partially Ordered Knapsack} with ${\mathcal L}=\left((s_1,
    w_1), \dots, (s_r, w_r)\right)$, $s_i=w_i$, $w_i\le r$ for all
  $i\in [r]$ and $b=t$, we construct an
  election %
  as follows.  For each $i\in[r]$, let $C_i=\{c_i^1, \dots,
  c_i^{w_i}\}$ and set $C= \bigcup_{i\in[r]} C_i$.  We construct the
  set of voters $V$ and the voters' preferences so that the majority
  graph of the resulting election $(C, V)$ has the following
  structure:
  \begin{itemize}
  \item[(1)] for each $i\in[r]$ the induced subgraph on $C_i$ is a
    strongly connected tournament;
  \item[(2)] for each $(i, j)\in A$ there is an edge from each
    candidate in $C_i$ to each candidate in $C_j$;
  \item[(3)] there are no other edges.
  \end{itemize}
  Using McGarvey's theorem, we can ensure that the number of voters
  $|V|$ is polynomial in $|C|$; as we have $w_i\le r$ for all $i\in
  [r]$, it follows that both the number of voters and the number of
  candidates are polynomial in $r$. Finally, we let the target
  committee size $k$ be equal to the knapsack size~$t$.

  Let $I$ be a witness that $\langle{\mathcal L}, b, t, \Gamma
  \rangle$ is a `yes'-instance of {\sc Partially Ordered Knapsack}.
  Then the set of candidates $S = \bigcup_{i\in I}C_i$ is a weakly
  Gehrlein-stable committee of size $k$: by construction, $|S|=\sum_{i\in
    I}|C_i| = \sum_{i\in I}w_i = t$, and the partial order constraints
  ensure that $S$ has no incoming edges in the weighted majority graph
  of $(C, V)$.

  Conversely, suppose that $S$ is a weakly Gehrlein-stable committee of size
  $k$ for $(C, V)$.  Note first that for each $i\in [r]$ it holds that
  $C_i\cap S\neq\emptyset$ implies $C_i\subseteq S$. Indeed, if we
  have $c\in C_i\cap S$, $c'\in C_i\setminus S$ for some $i\in[r]$ and
  some $c, c'\in C_i$ then in the weighted majority graph of $(C, V)$
  there is a path from $c'$ to $c$. This path contains an edge that
  crosses from $C_i\setminus S$ into $C_i\cap S$, a contradiction with
  $S$ being a weakly Gehrlein-stable committee. Thus, $S=\bigcup_{i\in I}C_i$
  for some $I\subseteq [r]$, and we have $\sum_{i\in I}w_i=\sum_{i\in
    I}|C_i|=k = t$. Moreover, for each directed edge $(i, j)\in A$
  such that $C_j\subseteq S$ we have $C_i\subseteq S$: indeed, the
  weighted majority graph of $(C, V)$ contains edges from candidates
  in $C_i$ to candidates in $C_j$, so if $C_i\not\subseteq S$, at
  least one of these edges would enter $S$, a contradiction with $S$
  being a weakly Gehrlein-stable committee. It follows that $I$ is a witness
  that we have started with a `yes'-instance of {\sc Partially Ordered
    Knapsack}.
\end{proof}

We obtain a similar result for locally stable committees.

\begin{theorem}\label{thm:ls-hard}
  Given an election $E = (C, V)$ and a target committee size $k$, with
  $k\le |C|$, it is {\em NP}-complete to decide if there exists a
  locally stable committee of size $k$ for $E$.
\end{theorem}
\begin{proof}
  It is easy to see that this problem is in NP: given an election $(C,
  V)$ together with a target committee size $k$ and a committee $S$
  with $|S|=k$, we can check for each $c\in C\setminus S$ whether
  there exist at least $\left\lfloor\frac{n}{k}\right\rfloor +1$
  voters who prefer $c$ to each member of $S$.

  To prove NP-hardness, we reduce from {\sc 3-Regular Vertex
    Cover}. Recall that an instance of {\sc 3-Regular Vertex Cover} is
  given by a $3$-regular graph $G = (V, E)$ and a positive integer
  $t$; it is a `yes'-instance if $G$ admits a vertex cover of size at
  most $t$, i.e., a subset of vertices $V'\subseteq V$ with $|V'|\le
  t$ such that $\{\nu, \nu'\}\cap V'\neq\emptyset$ for each $\{\nu,
  \nu'\}\in E$.  This problem is known to be
  NP-complete~\cite{gar-joh:b:int}.

  Consider an instance $(G, t)$ of {\sc 3-Regular Vertex Cover} with
  $G=(V, E)$, $V=\{\nu_1, \dots, \nu_r\}$. Note that we have
  $|E|=1.5r$, and we can assume that $t<r-1$, since otherwise $(G, t)$
  is trivially a `yes'-instance.
  Given $(G, t)$, we construct an
  election %
  as follows. We set $C=V\cup X \cup Y \cup Z$, where $X=\{x_1, \dots,
  x_{1.5r}\}$, $Y=\{y_1, \dots, y_{1.5r}\}$, $Z = \{z_1, \dots,
  z_{1.5r}\}$.  For each edge $\{\nu, \nu'\}\in E$ we construct one
  voter with preferences $\nu\succ \nu'\succ \cdots$ and one voter with
  preferences $\nu'\succ \nu \succ \cdots$; we refer to these voters as
  the {\em edge voters}.  Also, for each $j \in [1.5r]$ we construct
  two voters with preferences $x_j\succ y_j\succ z_j\succ \cdots$, two
  voters with preferences $y_j\succ z_j \succ x_j\succ\cdots$, and two
  voters with preferences $z_j\succ x_j \succ y_j\succ\cdots$; we refer
  to these voters as the {\em xyz-voters}.  We set $k = t+3r$. Note
  that the number of voters in our instance is $n = 2|E| + 6\cdot 1.5r
  = 12r$. Thus, using the fact that $0<t<r-1$, we can bound
  $\frac{n}{k+1}$ as follows:
  $$
  \frac{n}{k+1} > \frac{12r}{4r}= 3, \qquad\text{and}\qquad \frac{n}{k+1} <
  \frac{12r}{3r}= 4.
  $$
  Thus the Droop quota is $q = \left\lfloor\frac{n}{k+1}\right\rfloor
  + 1 = 4$.

  Now, suppose that $V'$ is a vertex cover of size at most $t$; we can
  assume that $|V'|$ is exactly $t$, as otherwise we can add arbitrary
  $t-|V'|$ vertices to $V'$, and it remains a vertex cover.  Then $S =
  V'\cup X\cup Y$ is a locally stable committee of size $|S|=t+2\cdot
  1.5r= t+3r$. Indeed, for each voter one of her top two candidates is
  in the committee (for edge voters this follows from the fact that
  $V'$ is a vertex cover and for xyz-voters this is immediate from the
  construction), so local stability can only be violated if for some
  candidate $c\not\in S$ there are at least $q = 4$ voters who rank
  $c$ first.  However, by construction each candidate is ranked first
  by at most three voters.

  Conversely, suppose that $S$ is a locally stable committee of size
  $t+3r$.  The argument in Example~\ref{ex:nols} shows that $|S\cap
  \{x_j, y_j, z_j\}|\ge 2$ for each $j=1, \dots, 1.5r$.
  Hence, $|S\cap V|\le t$.  Now, suppose that $S\cap V$ is not a
  vertex cover for $G$. Consider an edge $\{\nu, \nu'\}$ with $\nu,
  \nu'\not\in S$.  Since $G$ is 3-regular, there are three edge voters
  who rank $\nu$ first; clearly, these voters prefer $\nu$ to each
  member of $S$. Moreover, there is an edge voter whose preference
  order is $\nu' \succ \nu\succ \dots$; this voter, too, prefers $\nu$
  to each member of $S$. Thus, we have identified four voters who
  prefer $\nu$ to $S$, a contradiction with the local stability of
  $S$.  This shows that $S\cap V$ is a vertex cover for $G$, and we
  have already argued that $|S\cap V|\le t$.
\end{proof}

As we have observed in the proof of Theorem~\ref{thm:ls-hard}, it is possible to verify in
polynomial time that a given committee is locally stable. This is not
the case for full local stability, as we the following theorem shows.

\begin{theorem}
Given an election $E=(C, V)$ and a committee $S$,
it is {\em coNP}-complete to decide whether $S$ provides full local stability for $E$.
\label{thm:fls-coNP}
\end{theorem}
\begin{proof}
To see that this problem is in coNP, note that a certificate for a `no`-instance is an integer $\ell \in [k]$, a set $V^*$ of voters with $|V^*|=\left\lfloor \frac{\ell n}{|S|+1}\right\rfloor+1$ and a set $T$ of candidates with $|T|=\ell$ such that voters in $V^*$ prefer $T$ to $S$.

For hardness, we reduce from the NP-complete {\sc Multicolored Clique} problem \cite{fellows2009parameterized} to the 
complement of our problem.
An instance of {\sc Multicolored Clique} is given by an undirected graph $G=(U, {\mathcal E})$, 
a positive integer $s$, and a mapping (coloring) $g\colon U\to [s]$;
it is a `yes`-instance if there exists a set of vertices $\{u_1,\dots,u_s\}\subseteq U$ with 
$g(u_i)=i$ for every $i\in[s]$ such that $\{u_1,\dots,u_s\}$ forms a clique in $G$.
We write $U_a$ to denote the neighborhood of a vertex $a\in U$, i.e., $N_a=\{b\in U\mid \{a,b\}\in {\mathcal E}\}$ and we write $U_i$ to denote all $i$-colored vertices, i.e., $U_i=\{u\in U:g(u)=i\}$.

We have to make a few additional assumptions, all of which do not impact the hardness of {\sc Multicolored Clique}:
First, we assume that $s>2$, clearly hardness remains to hold.
Further, we assume without loss of generality that $s^2$ divides $|U|$ and that that $|U_i|=\nicefrac{|U|}{s}$; this can be achieved by adding disconnected vertices.
Finally, we assume that candidates of the same color are not connected.

We construct an election as follows:
Let $S=\{w_1,w_2,\dots,w_{s+2}\}$ and $C=U \cup S$.
We refer to candidates in $U$ as vertex candidates and we say that $u$ is an $i$-colored candidate if $g(u)=i$.
We create a voter $v_a$ for each vertex $a \in U$;
this voter's preferences are \[a\succ N_a \succ w_1 \succ \dots \succ w_{s+1} \succ U\setminus (N_a\cup \{a\}) \succ w_{s+2},\] where sets are ordered arbitrarily.
Let $V_U=\{v_a\mid a\in U\}$.
Furthermore, for every $i,j\in[s]$ we create a set of voters $V_i^j$ of size $|V_i^j|=(s+1)\cdot \frac{|U|}{s^2}$.
For an integer $z$, let $\overline{z}$ denote the number $(z \mod s)+1$.
Voters in $V_i^j$ have preferences of the form
\begin{align*}
w_{s+1} & \succ U_{\overline{i}}\succ w_{\overline{j}} \succ U_{\overline{i+1}} \succ w_{\overline{j+1}} \succ\dots\\
&\succ U_{\overline{i+s-1}}\succ w_{\overline{j+s-1}}\succ w_{s+2}.
\end{align*}
Let $\bar V=\bigcup_{i,j\in[s]} V_i^j$.
Finally, let $V'$ contain $|U|+s+1$ voters of the form $S\succ U$.
We set $V= V_U\cup \bar V\cup V'$ and have $|V|=(s+3)\cdot |U|+s+1$.
Thus, for $\ell$-local stability we have a quota of $\left\lfloor \frac{\ell\cdot |V|}{|S|+1} \right\rfloor+1=\ell\cdot |U|+\left\lfloor \frac{\ell\cdot (s+1)}{s+3} \right\rfloor+1$.
Note that we have $|V^*|\geq \ell\cdot |U|+\left\lfloor \frac{\ell\cdot (s+1)}{s+3} \right\rfloor+1$ if and only if $|V^*|> \ell\cdot |U|+ \frac{\ell\cdot (s+1)}{s+3}$; we will use this condition in the following proof.
Let us now prove that $S$ does not provide full local stability for $(C, V)$ if and only if $(U, {\mathcal E})$ has a clique of size $s$.

Let $U'$ be a multicolored clique of size $s$ in $G$, i.e., for all $i\in [s]$ it holds that $C_i\cap U'\neq\emptyset$.
We will show that $S$ violates $(s+1)$-local stability.
Let us consider $T=\{w_{s+1}\}\cup U'$; we claim that a sufficient number of voters prefers $T$ to $S$.
Note that $s$ voters corresponding to $U'$ prefer $T$ to $S$.
Furthermore, all voters in $\bar V$ prefer $T$ to $S$.
In total these are $s+s^2\cdot (s+1)\cdot \frac{|U|}{s^2} = (s+1)\cdot |U| +s$. We have to show that $(s+1)|U| +s > (s+1)|U|+\frac{(s+1)\cdot (s+1)}{s+3}$. This is equivalent to $s(s+3)>(s+1)^2$, which holds for $s \geq 2$. Hence $S$ is not $(s+1)$-locally stable and thus does not provide full local stability.

For the converse direction, let us make the following useful observation: if $T$ with $|T|=\ell$ contains an element that is ranked below the $\ell$-th representative of voter $v$, then $v$ does not prefer $T$ to $S$.
Now let us first show that $S$ provides $\ell$-local stability for all $\ell\in\{1,\dots,s,s+2\}$.
To see that $S$ provides $(s+2)$-local stability, let $T\subseteq C$ with $|T|=s+2$ and let $V^*\subseteq V$ be a set of the necessary size, i.e., $|V^*|> (s+2)\cdot |U|+ \frac{(s+2)\cdot (s+1)}{s+3}$. Hence $V^*$ has to contain voters from $V'$. But for them an improvement is not possible since for any $v\in V'$, $\best{}_{s+2}(v)=S$.

To see that $S$ is 1-locally stable, note that $|V_U|$ is lower than the quota for $\ell = 1$, $|V_U|< \left\lfloor \frac{|V|}{s+3} \right\rfloor+1$. Voters from $\bar V$ and from $V'$ have their top-ranked candidate in $S$. Hence no group of sufficient size can deviate.

For $1<\ell\leq s$, if $T$ does not contain $w_{s+1}$, then voters from $\bar V$ would not deviate. Since voter from $V'$ would also not deviate either, $V^*$ would be too small:
\begin{align*}
|V^*| \leq |V_U| = |U| < 2\cdot |U|+ \frac{2\cdot (s+1)}{s+3} \text{.}
\end{align*}
Hence we can assume that $w_{s+1}\in T$.
Then, however, voters from $V_U$ are excluded as $w_{s+1}$ is only their $(s+1)$-st representative.
Hence we only have to consider voters from $\bar V$.
Note that $T$ has to contain at least one vertex candidate, because otherwise $T\subseteq S$.
Since all $V_i^j$ are symmetric, we can assume without loss of generality that $U_1\cap T\neq \emptyset$, i.e., $T$ contains a $1$-colored vertex.
We distinguish whether $T\cap \{w_1,\dots, w_{s}\}$ is empty or not; in both cases we show that $|V^*|$ cannot have a sufficient size.

If $T\cap \{w_1,\dots, w_{s}\}\neq\emptyset$, we assume (again without loss of generality) that $w_1\in T$.
Observe that if $V_i^j$ prefers $T$ to $S$, then by construction of voters in $V_i^j$
it has to hold that
\begin{align}
T\subseteq \{w_{s+1}\}\cup \{w_{\overline{j}},\dots,w_{\overline{j+\ell-2}}\} \cup U_{\overline{i}}\cup \dots\cup U_{\overline{i+\ell-2}}.\label{eq:assumption2}
\end{align}
Since $w_1\in T$, condition~\eqref{eq:assumption2} implies that $1\in \{\overline{j}, \ldots, \overline{j+\ell-2}\}$.
Similarly, since $T$ contains a $1$-colored vertex, condition~\eqref{eq:assumption2} implies that $1\in \{\overline{i}, \ldots, \overline{i+\ell-2}\}$.
We see that for both $i$ and $j$ there are $\ell-1$ possible values. Similarly as before we infer that none of the voters from $V_U$ prefers $T$ to $S$---this is because $w_{s+1} \in T$, and $w_{s+1}$ is only the $(s+1)$-st representative of voters in $V_U$. Hence $T$ (with cardinality lower than $s+1$) cannot be desirable for them. Clearly, none of the voters from $V'$ prefers $T$ to $S$. 
Thus, it follows that
\begin{align*}
|V^*| &\leq (\ell-1)^2\cdot (s+1)\cdot \frac{|U|}{s^2} 
      \leq (\ell-1)\cdot (s-1)(s+1)\cdot \frac{|U|}{s^2} < \ell \cdot |U| \text{.}
\end{align*}
Hence the size of $V^*$ cannot be sufficiently large.

Now we consider the case that $T\cap \{w_1,\dots, w_{s}\}=\emptyset$ and hence $T\subseteq U\cup \{w_{s+1}\}$.
If $V_i^j$ prefers $T$ to $S$, then $T$ has to contain $w_{s+1}$, at least one element of $U_{\overline{i}}\cup \{w_{\overline{j}}\}$, at least two elements of $U_{\overline{i}}\cup U_{\overline{i+1}}\cup \{w_{\overline{j}},w_{\overline{j+1}}\}$, $\dots$, and at least $\ell-1$ elements of $U_{\overline{i}}\cup \dots\cup U_{\overline{i+\ell-2}}\cup \{w_{\overline{j}},\dots,w_{\overline{j+\ell-2}}\}$.
Since $T\subseteq U\cup \{w_{s+1}\}$, we can assume without loss of generality that $T\subseteq U_1\cup\dots\cup U_{\ell-1}\cup \{w_{s+1}\}$.
Also, without loss of generality, we can assume that $T$ contains a candidate from $U_1$. This implies that only voters from $V_1^j$ ($j$ arbitrary) may prefer $T$ to $S$.
Now:
\begin{align*}
|V_1^1\cup\dots\cup V_1^s|=s(s+1)\cdot \frac{|U|}{s^2}=|U|+\frac{|U|}{s^2}<2|U|\leq \ell|U|.
\end{align*}
Similarly as before, none of the voters from $V_U$ and $V'$ prefers $T$ to $S$.
Hence, also in this case, we have shown $V^*$ cannot be sufficiently large.
We conclude that $S$ satisfies $\ell$-local stability for $\ell\in\{2,\dots,s\}$.

We have established that $S$ provides $\ell$-local stability for all $\ell\in\{1,\dots,s,s+2\}$.
Hence, if $S$ fails full local stability, then it fails $(s+1)$-local stability.
Let $V^*\subseteq V$ and $T\subseteq C$ witness that $S$ is not $(s+1)$-locally stable.
First, let us show that $|V^*|\geq (s+1)\cdot |U|+ s$:
Since $V^*$ witness that $S$ is not $(s+1)$-locally stable, we know that $|V^*|> (s+1)\cdot |U|+ \frac{(s+1)\cdot (s+1)}{s+3}$. 
Note that $\frac{(s+1)^2}{s+3}=s-1+\frac{4}{s+3}$. Since $s>2$, $|V^*|\geq (s+1)\cdot |U|+s$.
Since $|V^*|\geq (s+1)\cdot |U|+ s$ and $V^*\cap V'=\emptyset$ (no improvement is possible for voters in $V'$), $V^*$ has to contain at least $s$ voters from $V_U$.
First, we show that $T\cap U$ contains a vertex of every color and $|T\cap U| = s$.
Then we are going to show that $|V^*\cap V_U|=s$.
We conclude the proof by showing that the corresponding vertices form a clique in $G$.

To show that $T\cap U$ contains a vertex of every color, let us first observe that $w_{s+1}\in T$; otherwise voters in $\bar V$ would not prefer $T$ over $S$ and so $V^*$ would not be of sufficient size.
Since $T\not\subseteq S$, there exists a $j\in[s]$ such that $w_j\notin T$.
Now assume towards a contradiction that $T\cap U$ contains no $i$-colored vertices.
Let $x$ denote the number of colors which are not used in $T$; by our assumption $x \geq 1$.
We are going to show that in this case $V^*$ is not of sufficient size:
If $T$ contains neither $i$-colored vertices nor $w_j$, then voters in $V_i^j$ do not prefer $T$ to $S$.
Thus, $|V^*|$ contains at most $(s+1)|U|-x(s+1)\cdot\frac{|U|}{s^2}$ voters from $\bar V$.
Next, for $i\in[s]$, if $T$ contains an $i$-colored vertex, say vertex $a$, then only one $i$-colored voter prefers $T$ to $S$ and that is $v_a$. This follows from the fact that $i$-colored vertices are not connected; hence $v_a$ is the only $i$-colored voter that ranks $a$ above $\{w_1,\dots,w_{s+1}\}$, which is a necessary requirement for $T$ (containing $a$) to be preferable to $S$.
If $T$ does not contain an $i$-colored vertex, then all $i$-colored voters may prefer $T$ to $S$; recall these are $\frac{|U|}{s}$ many.
We see that $|V^*|$ contains at most $x\cdot \frac{|U|}{s}+(s-x)$ voters from $V_U$. Further, $|V^*|$ contains
no voters from $V'$.
This yields an upper-bound on the total number of voters in $V^*$:
\begin{align*} 
|V^*| &\leq x\cdot \frac{|U|}{s}+(s-x) + (s+1)|U|-x(s+1)\cdot\frac{|U|}{s^2}
      <s + (s+1)|U| \text{,}
\end{align*}
which yields a contradiction.
Hence $T\cap U$ contains a vertex of every color.

Since $|V^*|\geq s + (s+1)|U|$, the set $V^*$ has to contain at least $s$ voters from $V_U$.
Observe that voter $v_a$ with $g(a)=i$ may only prefer $T$ to $S$ if $a\in T$. This follows from the already established facts that $T$ contains an $i$-colored vertex and, assuming this vertex is $a$, $v_a$ is the only $i$-colored voter ranking $a$ above $\{w_1,\dots,w_{s+1}\}$.
Furthermore, if $v_a$ prefers $T$ to $S$, it has to hold that $T\cap U\subseteq N(a)\cup \{a\}$.
Hence $T\cap U$ is a clique. As $T\cap U$ contains a vertex of every color, $T\cap U$ is a multicolored clique.
\end{proof}

\begin{corollary}
Given an election $E=(C, V)$ and a committee $S$,
it is {\em W[1]}-hard to decide whether $S$ provides full local stability for $E$ when parameterized by the committee size $k$.
\end{corollary}
\begin{proof}
The {\sc Multicolored Clique} problem is {\em W[1]}-hard \cite{fellows2009parameterized} and the reduction used in the proof of Theorem~\ref{thm:fls-coNP} is a parametrized reduction ($k=s+2$).
\end{proof}

We have not settled the complexity of finding a committee that
provides full local stability, but we expect this problem to be computationally hard as well. 
More precisely, it belongs to the second level of the
polynomial hierarchy (membership verification can be expressed as
``there exists a committee such that each possible deviation by each
group of voters is not a Pareto improvement for them,'' where both
quantifiers operate over objects of polynomial size);
we expect the problem to be complete for this complexity class.

\section{Conclusions and Research Directions}
\noindent
We have considered two generalizations of the notion of a Condorcet
winner to the case of multi-winner elections: the one proposed by
Gehrlein~\cite{Gehrlein1985199} and Ratliff~\cite{ratliff2003some} and
the one defined in this paper (but inspired by the works of Aziz et
al.~\cite{aziz:j:jr} and Elkind et
al.~\cite{elk-fal-sko-sli:c:multiwinner-rules}). We have provided
evidence that the former approach is very majoritarian in spirit and
is well-suited for shortlisting tasks (in particular, we
have shown that the objection based on weakly Gehrlein-stable rules
necessarily failing enlargement consistency does not apply to strongly
Gehrlein-stable rules). On the other hand, we have given arguments
that local stability may lead to diverse committees, whereas full
local stability may lead to committees that represent the voters
proportionally. (We use qualifications such as ``may lead'' instead of
``leads'' because, technically, (fully) local stable rules 
may behave arbitrarily on elections where (fully)
locally stable committees do not exist).

In our discussion, we have only very briefly mentioned rules
that are either Gehrlein-stable or locally stable. 
Many such rules have been defined in the literature~\cite{Kamwa2015b}, 
and these rules call for a more detailed study, 
both axiomatic and algorithmic. Our results
indicate that weakly Gehrlein-stable and locally stable rules 
are unlikely to be polynomial-time computable; it would be desirable 
to find practical heuristics or design efficient exponential algorithms.

\bibliographystyle{plain}
\bibliography{main}

\end{document}